\documentclass[11pt]{article}
\usepackage[margin=2.54cm]{geometry}
\usepackage[dvipsnames,usenames]{color}
\usepackage{amsfonts,amsmath,amssymb,amsthm,mathtools}
\usepackage{paralist}
\usepackage{algorithmic,algorithm}
\usepackage{bm}
\usepackage{xspace}
\usepackage{xspace,prettyref}
\usepackage{color}
\usepackage{mathbbol}
\usepackage[pagebackref,letterpaper=true,colorlinks=true,pdfpagemode=none,urlcolor=blue,linkcolor=blue,citecolor=BrickRed,pdfstartview=FitH]{hyperref}

\newcommand{\ignore}[1]{}

\newtheorem{theorem}{Theorem}

\newtheorem{lemma}[theorem]{Lemma}
\newtheorem{claim}[theorem]{Claim}

\newtheorem{corollary}[theorem]{Corollary}

\theoremstyle{definition}

\newcommand{\EX}{\hbox{\bf E}}
\newcommand{\pr}{{\rm Pr}}
\newcommand{\var}{{\rm var}}

\def\eps{\varepsilon}
\def\bar{\overline}

\def\ceil#1{\lceil {#1} \rceil}

        {\hspace*{\fill}$\Box$\par}

\newcommand{\cE}{{\cal E}}
\newcommand{\cF}{{\cal F}}

\newcommand{\bone}{\mathbb{1}}

\newcommand{\Sec}[1]{\hyperref[sec:#1]{\S\ref*{sec:#1}}} 
\newcommand{\Eqn}[1]{\hyperref[eq:#1]{(\ref*{eq:#1})}} 
\newcommand{\Fig}[1]{\hyperref[fig:#1]{Fig.\,\ref*{fig:#1}}} 
\newcommand{\Tab}[1]{\hyperref[tab:#1]{Tab.\,\ref*{tab:#1}}} 
\newcommand{\Thm}[1]{\hyperref[thm:#1]{Theorem\,\ref*{thm:#1}}} 
\newcommand{\Lem}[1]{\hyperref[lem:#1]{Lemma\,\ref*{lem:#1}}} 
\newcommand{\Prop}[1]{\hyperref[prop:#1]{Prop.~\ref*{prop:#1}}} 
\newcommand{\Cor}[1]{\hyperref[cor:#1]{Corollary~\ref*{cor:#1}}} 
\newcommand{\Def}[1]{\hyperref[def:#1]{Definition~\ref*{def:#1}}} 
\newcommand{\Alg}[1]{\hyperref[alg:#1]{Alg.~\ref*{alg:#1}}} 
\newcommand{\Ex}[1]{\hyperref[ex:#1]{Ex.~\ref*{ex:#1}}} 
\newcommand{\Clm}[1]{\hyperref[clm:#1]{Claim~\ref*{clm:#1}}} 

\def\poly{{\sf poly}}

\newcommand{\davg}{\overline{d}}
\newcommand{\hatd}{\widehat{d}}

\newcommand{\const}{10}

\newcommand{\adderr}{\Delta}
\newcommand{\mulerr}{\eps}

\newcommand{\wt}{\hbox{wt}}

\newcommand{\heav}{{\tt heavy}}
\newcommand{\triest}{{\tt estimate}}
\newcommand{\twoapp}{{\tt 2-approx}}

\author{C. Seshadhri
 \\ {\tt scomandu@ucsc.edu}\\
University of California, Santa Cruz
}

\title{A simpler sublinear algorithm for approximating the triangle count}
\date{}

\begin{document}

\maketitle

\begin{abstract} A recent result of Eden, Levi, and Ron (ECCC 2015) provides a sublinear time
algorithm to estimate the number of triangles in a graph. Given an undirected graph $G$,
one can query the degree of a vertex, the existence of an edge between vertices,
and the $i$th neighbor of a vertex. Suppose the graph has $n$ vertices, $m$ edges, and $t$
triangles. In this model, Eden et al provided a $O(\poly(\eps^{-1}\log n)(n/t^{1/3} + m^{3/2}/t))$
time algorithm to get a $(1+\eps)$-multiplicative approximation for $t$, the triangle count.
This paper provides a simpler algorithm with the same running time (up to differences
in the $\poly(\eps^{-1}\log n)$ factor) that has a substantially simpler analysis.
\end{abstract}

\section{Introduction}

Counting the number of triangles in a graph is a fundamental algorithmic problem. It has been
studied widely in theory and well as practice. With the recent study of complex networks
and massive real-world graphs, triangle counting is a key operation in graph
analysis for bioinformatics, social networks, community analysis, and graph modeling~\cite{HoLe70,Co88,EcMo02,Po98,Milo2002,Burt04,BeBoCaGi08,FoDeCo10,BerryHLP11,SeKoPi11}.

There has been much study on triangle counting by the theoretical computer science community, where the primary hammer is 
fast matrix multiplication~\cite{ItRo78,AlYuZw97,BjPa+14}. On the more applied side, there is a plethora
of provable and practical algorithms that employ clever sampling methods for
approximate triangle counting~\cite{ChNi85,ScWa05,ScWa05-2,Ts08,TsKaMiFa09,KoMiPeTs10,Av10,ChCh11,SuVa11,TsKoMi11,ArKhMa12,SePiKo13,TaPaTi13}. Triangle counting has also been a popular problem
for the streaming setting~\cite{BaKuSi02,JoGh05,BuFrLeMaSo06,AhGuMc12,KaMeSaSu12,JhSePi13,PaTaTi+13,TaPaTi13,AhDuNe+14}.

Yet all these algorithms always read the entire graph. 
A recent result of Eden, Levi, and Ron~\cite{EdLe+15} provided the first truly \emph{sublinear} algorithm.
For a graph with $n$ vertices, $m$ edges, and $t$ triangles, their
running time is $O(\poly(\eps^{-1}\log n)(n/t^{1/3} + m^{3/2}/t))$.
Observe that when $t = \Omega(m^{1/2+\delta})$ (for constant $\delta > 0$), this algorithm does not even read the entire graph.
They also prove a lower bound, tight up to the $\poly(\eps^{-1}\log n)$ factor.
The algorithm and analysis are quite intricate, and require a number of ideas.
This paper provides a simpler algorithm with a simpler analysis.

Let us explain the formal result. The input is an undirected graph $G = (V,E)$, stored as an adjacency list.
We set $V = [n]$, where $n$ is known to the algorithm. The following queries are allowed:

\medskip
\begin{asparaitem}
    \item Degree queries: for any $v \in V$, we can get $d_v$, the degree of $v$.
    \item Neighbor queries: for any $v \in V$ and $i \leq d_v$, we can get the $i$th neighbor of $v$.
    \item Edge queries: for any $u,v \in V$, we can check if $(u,v) \in E$.
\end{asparaitem}
\medskip

We denote the number of edges by $m$ and the number of triangles by $t$. Our main result
is the following (identical to that of~\cite{EdLe+15}).

\begin{theorem} \label{thm:main-result} There exists an algorithm with the following guarantee.
Given query access to graph $G$ and an approximation parameter $\eps > 0$, the algorithm
outputs an estimate for $t$ in the range $[(1-\eps)t,(1+\eps)t]$ with probability at least $>2/3$.  
The expected running time of the algorithm is $O(\poly(\eps^{-1}\log n)(n/t^{1/3} + m^{3/2}/t))$.
\end{theorem}

We can easily boost the success probability by repeating and taking the median estimate.
The query complexity can be made $O(\poly(\eps^{-1}\log n)\min(m, n/t^{1/3} + m^{3/2}/t))$.
As we shall see, the algorithm begins by quickly obtaining an accurate estimate of $m$. If the algorithm
ever makes more queries than $m$, we can pause, query the entire graph and store it. Now, we run
the remainder of the algorithm on this stored version.

\subsection{Previous work} \label{sec:prev}


For the sake of clarity, we suppress any dependences on the approximation parameter $\eps$
or on $\log n$ using the notation $O^*(\cdot)$.
The result of Eden, Levi, and Ron~\cite{EdLe+15} (denoted ELR) can be traced back to ideas of Feige~\cite{Fe02} and Goldreich and Ron~\cite{GoRo08},
who provided sublinear algorithms for estimating the average degree. Feige shows that
averaging the degree of $O^*(n/\sqrt{m})$ uniform random vertices suffices to give a $(2+\eps)$-approximation
of the average degree. His primary tool was a new tail bound for sums of independent
random variables with unbounded variance. Goldreich and Ron subsequently give (Feige also observes
this in the journal version) a substantially simpler proof, and the cost of a worse dependence
on $\eps$. Extending the analysis, they provide a $(1+\eps)$-approximation with the same query complexity
when the algorithm is also allowed neighbor queries. Gonen, Ron, and Shavitt~\cite{GoRo+11} build on these
ideas to count stars, and ELR goes further to count triangles. 

We note that there is a significant
leap required to perform triangle counting. Since the query model allows the algorithm
to directly get degrees, it is plausible that one can estimate moments of the degree
sequence. Triangle counting is a much harder beast, since mere knowledge of degrees
provides no (obvious) help.

The high level approach in ELR, building on Goldreich-Ron and Gonen et al,
is vertex bucketing. We bin the vertices in $O(\log n)$ buckets,
by rounding their triangle count (or degree, in previous work) down to the nearest power of $(1+\eps)$. With respect
to the buckets, each edge or triangle has only $\poly(\log n)$ different possibilities, or ``types".
So we can apply concentration bounds for each type of edge/triangle, take a union
bound over the different types.
Of course, the challenge is now in counting triangles of a given type. There are numerous
issues caused by small buckets and determining the type of a triangle, which lead to a fairly complex
algorithm and analysis.

\subsection{Main ideas} \label{sec:idea}

Our approach goes via a different route. Start with the exact triangle counter of Chiba and Nishizeki~\cite{ChNi85}.
Simply enumerate over all edges, and count the number of triangles $t_e$ containing each edge $e$. For edge $e=(u,v)$,
it suffices to compute the intersection of neighbor lists of $u$ and $v$. This can be done in $\min(d_u, d_v)$ time,
leading to an overall bound of $\sum_{e = (u,v)} \min(d_u,d_v)$. By bounds on the arboricity (or degeneracy) of graphs,
this expression can be bounded by $O(m^{3/2})$. How about implementing this algorithm using sampling?

Pick a u.a.r. (uniform at random) edge $e$ and let $u$ denote the endpoint of lower degree. Pick
a u.a.r. neighbor $w$ of $u$. If $e$ and $w$ form a triangle, set $X = d_u$. Else, $X = 0$. The expected
value of $X$ is precisely $3t/m$, so one hopes that $O(m/t)$ samples suffices to get concentration.
But the variance of $X$ could be too large, so we simply use more samples for ``high variance edges".
After picking $e$ with endpoint $u$, if $d_u$ is large, we sample more neighbors $w_1, w_2, \ldots$
and take the average of the corresponding $X$ values. We set the number of these neighbors to be $d_u/\sqrt{m}$,
and can show that $\var(X) \leq \sqrt{m} \EX[X]$. So $O(m^{3/2}/t)$ samples suffice for concentration.
But the time per sample has potentially gone up from constant to expected $m^{-1}\sum_{e = (u,v)} \min(d_u,d_v)/\sqrt{m}$.
By Chiba-Nishizeki's bound, this is still $O(1)$ and we are done!

Our problem is not yet solved,
because we need to sample a uniform random \emph{edge}, a query that is not allowed.
We can sample some set $S$ of vertices, query all their degrees, and can
sample a u.a.r. edge incident to $S$ in $O(1)$ time.
We could perform the above algorithm
on edges incident to $S$, but we required the number of triangles incident to $S$ to concentrate well. This is 
suspiciously similar to the average degree questions of Feige and Goldreich-Ron. Using a stripped-down
version of Feige's analysis, we can argue that sampling $O(n/t^{1/3})$ vertices is enough, with a painful caveat. 
This only provides a $3$-approximation to $t$ (analogous to Feige's $2$-approximation for average degree).
To get down to $(1+\eps)$, we need to sample a uniform triangle incident to $S$, 
determine the number triangles incident to \emph{these} vertices, and then weight the triangle accordingly.
(ELR also perform the sampling of $n/t^{1/3}$ vertices and explicitly move from the $3$ to $(1+\eps)$-approximation.
In this analysis, we go straight to the $(1+\eps)$-approximation.)
The number of triangles incident to a vertex $v$ can be estimated by the algorithm in the previous paragraph,
by sampling edges incident to $v$. With some care, all of this can be put together in $O^*(n/t^{1/3} + m^{3/2}/t)$
time.

As an aside, we also give a simpler analysis of a $(1+\eps)$-approximation $O^*(n/\sqrt{m})$ algorithm for the average degree,
first shown by Goldreich-Ron~\cite{GoRo08}. We defer this to the end of the paper.

\section{Preliminaries}

We use $d_v$ for the degree of vertex $v$, $\delta(v)$ for the set of edges incident to $v$,
and $\Gamma(v)$ for the neighborhood of $v$.
We use $t_e$ for the number of triangles incident to edge $e$, and set $t_v = \sum_{e \in \delta(v)} t_e$.
Note that the latter is twice the number of triangles incident to $v$.
The set of triangles incident to $e$ is $T_e$.
We use $c, c_1, \ldots$ to denote sufficiently large constants.

Our initial description of the algorithm will use the value of $m$ and $t$ to decide
how much to sample. This (somewhat circular) assumption is easily removed by doing
a geometric search on $m$ and $t$, and is explained at the end of our proof.
We use $\eps$ to denote the approximation parameter.

\section{Heavy and light vertices} \label{sec:heavy}

Roughly speaking, a \emph{heavy} vertex is one with either high degree or many 
triangles. By ``high", we mean in the top $t^{1/3}$ values.
For the main algorithm, it will be important to distinguish such heavy vertices
efficiently. 

\medskip
\fbox{
\begin{minipage}{0.9\textwidth}
{\bf \heav$(v)$}

\smallskip
\begin{compactenum}
    \item If $d_v > 2m/(\eps t)^{1/3}$, output \textbf{heavy}.
    \item Repeat for $i = 1,2,\ldots,c\log n$:
    \begin{compactenum}
        \item Repeat for $j = 1,2,\ldots,(4/\eps^2)(m^{3/2}/t) = s$:    
        \begin{compactenum}
            \item Select u.a.r. edge $e \in \delta(v)$, and let $u$ be endpoint with smaller degree.
            \item Repeat for $k = 1,2,\ldots,\ceil{d_u/\sqrt{m}}$:
            \begin{compactenum}
                \item Pick u.a.r. neighbor $w$ of $u$.
                \item If $e$ with $w$ forms a triangle, set $Z_k = d_u$, else $Z_k = 0$.
            \end{compactenum}
            \item Set $Y_j = \sum_k Z_k/\ceil{d_u/\sqrt{m}}$.
        \end{compactenum}
        \item Set $X_i = d_v \sum_j Y_j/s$.
    \end{compactenum}
    \item If median of $X_i$s is greater than $t^{2/3}/\eps^{1/3}$, output \textbf{heavy}, else output \textbf{light}.
\end{compactenum}
\end{minipage}}

\medskip
We have three nested loops, with loop variables $i,j,k$ respectively. We reference these
as ``iteration $i$", ``iteration $j$", and ''iteration $k$".

\begin{lemma} \label{lem:x} For any iteration $i$, $\pr[|X_i - t_v| > \eps \max(t_v, td_v/m)] < 1/4$. 
\end{lemma}

\begin{proof} Fix an iteration $j$, and let $e_j$ denote the edge chosen in the $j$th iteration, with $u_j$ 
as the smaller degree endpoint. We use $\cE_j$ to denote the event of $e_j$ being chosen.
The probability of finding a triangle in any iteration $k$
is $t_{e_j}/d_{u_j}$. Hence, $\EX[Z_k | \cE_j] = (t_{e_j}/d_{u_j})\cdot d_{u_j} = t_e$, and $\var(Z_k | \cE_j) \leq \EX[Z^2_k | \cE_j] \leq d_{u_j} \EX[Z_k | \cE_j]$.
By linearity of expectation, $\EX[Y_j | \cE_j] = t_e$ and by independence, $\var(Y_j | \cE_j) \leq d_{u_j} \EX[Y_j]/\ceil{d_{u_j}/\sqrt{m}} \leq \sqrt{m} \EX[Y_j | \cE_j]$.
The conditioning can be removed to yield $\EX[Y_j] = \sum_{e \in \delta(v)} t_e/d_v = t_v/d_v$
and $\var(Y_j) \leq \sqrt{m} \EX[Y_j]$.

By Chebyshev's inequality on $\overline{Y} = \sum_j Y_j/s$, 
$\pr[|\overline{Y} - t_v/d_v| > \eps \max(t_v/d_v, t/m)]$ is at most
$$ \frac{\var(\overline{Y})}{\eps^2 \max(t_v/d_v,t/m)^2} \leq \frac{\sqrt{m}(t_v/d_v)}{\eps^2(4/\eps^2)(m^{3/2}/t) \cdot(t_v/d_v)\cdot(t/m)} = 1/4 $$
\end{proof}

\begin{lemma} \label{lem:heavy} The following hold with probability $> 1-1/n$ over all calls of \heav.
If $v$ is declared light, then $t_v \leq 2t^{2/3}/\eps^{1/3}$. If $v$ is declared heavy, then $d_v > 2m/(\eps t)^{1/3}$ or 
$t_v > t^{2/3}/2\eps^{1/3}$.
\end{lemma}

\begin{proof}  It is more convenient to prove the contrapositive statements.
Obviously, if $d_v > 2m/(\eps t)^{1/3}$, then $v$ is heavy. So assume
that $d_v \leq 2m/(\eps t)^{1/3}$. Then, $td_v/m \leq 2t^{2/3}/\eps^{1/3}$.
Suppose $t_v > 2t^{2/3}/\eps^{1/3}$. By \Lem{x}, for any iteration $i$, $\Pr[|X_i - t_v| > \eps t_v] < 1/4$.
By a standard Chernoff bound, the median of the $c\log n$ $X_i$s will be greater than $t^{2/3}/\eps^{1/3}$
with probability at least $1-1/n^2$, and $v$ is declared heavy.

Suppose $t_v \leq t^{2/3}/2\eps^{1/3}$. By \Lem{x}, $\Pr[|X_i - t_v| > \eps (2t^{2/3}/\eps^{1/3})] < 1/4$.
By Chernoff again, the median will be less than $t^{2/3}/\eps^{1/3}$ with probability
at least $1-1/n^2$, and $v$ is light. A union bound over all $v$ completes the argument.
\end{proof}

Obviously, there is an upper bound on the number of vertices with either high degree or many
triangles.
\begin{corollary} \label{cor:heavy} The number of heavy vertices is at most $3(\eps t)^{1/3}$.
\end{corollary}

The following is where the degeneracy bounds come into play. We give a direct, self-contained proof,
but note the connection to Chiba-Nishizeki's bound.

\begin{lemma} \label{lem:xtime} The expected runtime of \heav$(v)$ is $O^*(m^{3/2}/t)$. 
\end{lemma}

\begin{proof} It suffices to argue that the expected time to generate a single sample of $Y_j$
is $O(1)$.
Our query model allows for selecting a u.a.r. edge in $\delta(v)$ by a single
query. If $d_v \leq \sqrt{m}$, then the degree of smaller endpoint for any $e \in \delta(v)$
is at most $\sqrt{m}$. So a sample is clearly generated in $O(1)$ time. Suppose $d_v > \sqrt{m}$.
If edge $e = (v,w)$ is sampled, the runtime is $O(1 + \min(d_v,d_w)/\sqrt{m})$. Hence, the
expected runtime to generate $Y_j$ is, up to constant factors, at most:
\begin{eqnarray*}
    d^{-1}_v \sum_{w \in \Gamma(v)} (1+m^{-1/2}\min(d_v,d_w))
    & = & 1 + (\sqrt{m}d_v)^{-1} \sum_{\substack{w \in \Gamma(v) \\ d_w \leq d_v}} d_w +(\sqrt{m}d_v)^{-1} \sum_{\substack{w \in \Gamma(v) \\ d_w > d_v}} d_v \\
    & \leq & 1 + m^{-1} \sum_w d_w + m^{-1/2} |\{w | d_w > d_v\}| \\
    & \leq & 3 + m^{-1/2}(2m/d_v) = O(1)
\end{eqnarray*}
\end{proof}

We will assume that the random coins used by \heav{} are fixed in advance, and that heavy/light
vertices satisfy the  condition of \Lem{heavy}. This allows us to treat the heavy/light labeling
as deterministic. By a union bound, it only adds $1/n$ to the errors in all subsequent probability
bounds.

\subsection{The utility of heavy vertices}

Let $L$ and $H$ denote the set of light and heavy vertices. 
For every triangle $\Delta$, we associate a weight depending on the number of light
vertices in $\Delta$. We set $\wt(\Delta) = 0$ if $\Delta$ has no light vertices,
and $\wt(\Delta) = 1/2\ell$, if $\Delta$ has $\ell \neq 0$ light vertices.

\begin{lemma} \label{lem:H} $\sum_{v \in L} \sum_{e \in \delta(v)} \wt(T_e) \in [t(1-9\eps),t]$.
\end{lemma}

\begin{proof} Define indicator $\chi(e,\Delta)$ for triangle $\Delta$ containing edge $e$. 
Consider a triangle $\Delta$ that contains $\ell \neq 0$ light vertices.
Then $\sum_{v \in L} \sum_{e \in \delta(v)} \chi(e,\Delta)$ is exactly $2\ell$, which is $1/\wt(\Delta)$.
If $\ell = \wt(\Delta) = 0$, this expression is $0$.
By interchanging summations,
$$ \sum_{v \in L} \sum_{e \in \delta(v)} \wt(T_e) = \sum_{\Delta} \wt(\Delta) \sum_{v \in L} \sum_{e \in \delta(v)} \chi(e,\Delta)
= t - |\{\Delta | \wt(\Delta) = 0\}|.$$
But the number of triangles with weight $0$ is at most ${|H|\choose 3}$, which by \Cor{heavy} is at most $9\eps t$.
\end{proof}

We define $\wt(T) = \sum_{\Delta \in T} \wt(\Delta)$ for any set $T$ of triangles.
Abusing notation, define $\wt(v) = \sum_{e \in \delta(v)} \wt(T_e)$ for $v \in L$,
and $\wt(v) = 0$ for $v \in H$.

\begin{theorem} \label{thm:vert} Let $s \geq (c\log(n/\eps)/\eps^3) n/t^{1/3}$. Sample $s$ u.a.r. vertices $v_1, v_2,\ldots, v_s$.
Then $\EX[\sum_{i \leq s} \wt(v_i)/s] \in [t(1-9\eps)/n,t/n]$ and $\Pr[\sum_{i\leq s} \wt(v_i)/s < t(1-10\eps)/n] < \eps^2/n$.
\end{theorem}

\begin{proof} The expectation holds by \Lem{H}. Let $Y$ denote the random variable $\sum_{i \leq s} \wt(v_i)/s$.
Note that $\wt(v) \leq t_v$,
which by \Lem{heavy} is at most $2t^{2/3}/\eps^{1/3}$.
A multiplicative Chernoff bound tells us that $\Pr[Y < \EX[Y](1-\eps)]
< \exp(-\eps^2s\EX[Y]/3(2t^{2/3}/\eps^{1/3}))$. Algebra on the exponent:
$$ \frac{\eps^2s\EX[Y]}{6t^{2/3}/\eps^{1/3}} \geq \frac{c\log(n/\eps)(n/\eps t^{1/3})\cdot t/(2n)}{6t^{2/3}/\eps^{1/3}}
= \Omega(\log(n/\eps))$$
\end{proof}

\section{The full algorithm} \label{sec:full}

We come to the main algorithm. It is convenient to define a procedure \triest{} uses the values of $m$
and $t$, and only has a lower tail bound. Later, we employ Markov and geometric search to get the bonafide 
algorithm that estimates $t$.

\medskip
\fbox{
\begin{minipage}{0.9\textwidth}
{\bf \triest}

\smallskip
\begin{compactenum}
    \item Sample $s_1 = c_1 \eps^{-3}\log (n/\eps) n t^{-1/3}$ u.a.r. vertices. Call this multiset $S$.
    \item Set up a data structure to sample vertices in $S$ proportional to degree.
    \item Repeat for $i = 1, 2, \ldots, s_2 = c_2 \eps^{-4}(\log^2n)  m^{3/2} t^{-1}$ times:
    
    \begin{compactenum}
        \item \label{step3} Sample $v \in S$ proportional to $d_v$ and sample u.a.r. $e \in \delta(v)$. Let $u$ be lower degree
        endpoint.
        \item If $d_u \leq \sqrt{m}$, set $R=1$ with probability $d_u/\sqrt{m}$ and set $R=0$ otherwise.
        If $d_u > \sqrt{m}$, set $R = \ceil{d_u/\sqrt{m}}$.
        \item Repeat for $j = 1,2,\ldots,R$:
        \begin{compactenum}
            \item Pick u.a.r. neighbor $w$ of $u$.
            \item If $e$ with $w$ does not form triangle, set $Z_j = 0$.
            \item If $e$ with $w$ forms triangle $\Delta$: call \heav{} for all vertices in $\Delta$.
            If $v$ is heavy, set $Z_j = 0$. Otherwise, set $Z_j = \max(d_u,\sqrt{m}) \wt(\Delta)$.
        \end{compactenum}
        \item Set $Y_i = \sum_j Z_j/R$. (If $R=0$, set $Y_i = 0$.)
    \end{compactenum}
    \item Output $X = n(\sum_{v \in S} d_v)(\sum_i Y_i)/s_1s_2$
\end{compactenum}
\end{minipage}}

\begin{theorem} \label{thm:estimate}$\EX[X] \in [t(1-2\eps),t]$ and $\Pr[X < t(1-20\eps)] < 3\eps/\log n$. 
\end{theorem}

\begin{proof} There are three ``levels" of randomness. First is the choice of $S$,
the second level is the choice of $e$ (Step~\ref{step3}), and finally the $Z_j$s.
For analyzing the randomness in any level, we condition on the previous levels.
It is helpful to break the proof up using claims.

\begin{claim} \label{clm:Y} Condition of some $S$ being chosen, and let $d_S = \sum_{v \in S} d_v$.
$\EX[Y_i|S] = d^{-1}_S \sum_{v \in S} \wt(v)$ and $\var(Y_i|S) \leq \sqrt{m} \EX[Y_i|S]$.
\end{claim}

\begin{proof} This is similar to the argument in \Lem{x}. Let vertex $v_i$ and edge $e_i$ with lower degree endpoint $u_i$ be chosen in iteration $i$.
We simply refer to this event by $\cE_i$, so the conditioning is over $S$ and $\cE_i$. (We use
$\bone_\cF$ is the indicator for event $\cF$.)

If $d_{u_i} \leq \sqrt{m}$, $\EX[Y | S, \cE_i] = (d_{u_i}/\sqrt{m})d^{-1}_{u_i}\bone_{v_i \in L}\sum_{\Delta \in T_{e_i}}
\sqrt{m}\cdot\wt(\Delta)$ $= \bone_{v_i \in L} \wt(T_{e_i})$. Since the maximum value of $Y$ in this case
is at most $\sqrt{m}$, $\var(Y | S, \cE_i) \leq \sqrt{m} \EX[Y | S, \cE_i]$.

We prove the same bound if $d_{u_i} > \sqrt{m}$. So $\EX[Z_j | S, \cE_i] = d^{-1}_{u_i} \bone_{v_i \in L} \sum_{\Delta \in T_{e_i}} d_{u_i} \wt(\Delta)
= \bone_{v_i \in L} \wt(T_{e_i})$. As before, $\var(Z_j | S, \cE_i) \leq d_{u_i} \EX[Z_j | S, \cE_i]$.
By linearity of expectation $\EX[Y_i | S, \cE_i] = \bone_{v_i \in L} \wt(T_{e_i})$. 
By independence of the $(Z_j | S, \cE_i)$ variables, $\var(Y_i | S, \cE_i) \leq \sqrt{m} \EX[Y_i | S, \cE_i]$.
We remove the conditioning on $\cE_i$:
$$ \EX[Y_i|S] = \sum_{v \in S} \frac{d_v}{d_S} \times \frac{1}{d_v} \sum_{e \in \delta(v)} \bone_{v_i \in L} \wt(T_{e_i})
= d^{-1}_S \sum_{v \in S} \sum_{e \in \delta(v)} \wt(T_{e_i}) \bone_{v_i \in L} = d^{-1}_S \sum_{v \in S} \wt(v)$$
We also have $\var(Y_i|S) \leq \sqrt{m} \EX[Y_i|S]$. 
\end{proof}

Hence, $\EX[X|S] = nd_S\EX[Y_1|S]/s_1 = n\sum_{v \in S}\wt(v)/|S|$.
By \Thm{vert}, the expectation over $S$ (which yields $\EX[X]$) is in $[t(1-9\eps),t]$.

We will call $S$ \emph{good} if $\sum_{v \in S}\wt(v)/s_1 \geq t(1-10\eps)/n$.
By \Thm{vert}, this happens with probability at least $1-\eps^2/n$. We call $S$ \emph{great}
if, in addition to being good, $d_S = \sum_{v \in S} d_v \leq s_1(2m/n)(\log n)/\eps$.
The expected value, over $S$, of $d_S$ is $s_1(2m/n)$. By the Markov bound and the union bound,
the probability that $S$ is great is at least $1-2\eps/\log n$.

We apply Chebyshev's inequality on $\overline{Y} = \sum_i Y_i/s_2$, conditioned
over $S$. We get that $\Pr[|\overline{Y}|S - \EX[\overline{Y}|S]| > \eps \EX[\overline{Y}|S]]$
is at most
$$ \frac{\var(\overline{Y}|S)}{\eps^2 \EX[\overline{Y}|S]^2} \leq \frac{\sqrt{m} \EX[\overline{Y}|S]}{\eps^2 (c_2 \eps^{-4}\log^2n) m^{3/2} t^{-1} 
\EX[\overline{Y}|S]^2} = \frac{\eps^2}{c_2(\log^2n) (m/t)\EX[\overline{Y}|S]}$$
Note that $\EX[\overline{Y}|S] = d^{-1}_S \sum_{v \in S} \wt(v)$, which for great $S$
is at least $(t/4n)(\log n/\eps)$. Hence, 
$\Pr[|\overline{Y}|S - \EX[\overline{Y}|S]| > \eps \EX[\overline{Y}|S]] \leq \eps/\log n$.
Conditioned on $S$, $X$ is just a scaling of $\overline{Y}$. So we
get $\Pr[|\overline{X}|S - \EX[\overline{X}|S]| > \eps \EX[\overline{X}|S]] \leq \eps/\log n$.
Note that $\EX[\overline{X}|S] = n\sum_{v \in S}\wt(v)/s_1$, which for great $S$
is at least $t(1-10\eps)$. Hence, for great $S$,
$\Pr[X|S < t(1-20\eps)/n] \leq \eps/\log n$. The probability of $S$ not being
great is at most $2\eps/\log n$. We apply the union bound to remove the conditioning, so
$\Pr[X < t(1-20\eps)/n] \leq 3\eps/\log n$.
\end{proof} 

\begin{theorem} \label{thm:time} The expected running time of \triest{} is $O^*(n/t^{1/3} + m^{3/2}/t)$.
\end{theorem}

\begin{proof} The sampling of $S$ is done in $O^*(n/t^{1/3})$ time. Let us compute the expected number
of triangles found. In iteration $i$, we are basically picking a u.a.r. edge of $G$. Conditioned 
on choosing edge $e$, the expected number of triangles found is at most $2(d_u/\sqrt{m})(t_e/d_u) = 2t_e/\sqrt{m}$.
Averaging over edge $e$, the expected number of triangles found in a single iteration is at most $2t/m^{3/2}$.
There are $O^*(m^{3/2}/t)$ iterations, leading to grand total of $O^*(1)$ expected triangles.
Thus, there are $O^*(1)$ expected calls to \heav, each taking $O^*(m^{3/2}/t)$ time by \Lem{xtime}.

The time required to generate the $Z_j$s is also $O^*(m^{3/2}/t)$
by an argument identical to that in the proof of \Lem{xtime}.
\end{proof} 

\begin{theorem} \label{thm:maintri} There exists a $O^*(n/t^{1/3} + m^{3/2}/t)$ algorithm
that provides an estimate for the triangle count in $[(1-\eps)t,(1+\eps)t]$ with probability at least $> 5/6$. 
\end{theorem}

\begin{proof} First, we need to estimate $m$.
This can be done with suitably high probability by the algorithm of Goldreich-Ron~\cite{GoRo08} in $O^*(n/\sqrt{m})$ time.
(We give an independent proof in the next section of \Thm{maindeg}, which gives the desired algorithm.)
Since $t \leq m^{3/2}$, this can be absorbed in the $O^*(n/t^{1/3})$ bound. 

We perform a geometric search for $t$, by guessing its value as $n^3, n^3/2, n^3/2^2, \ldots$.
We run the procedure of \Thm{estimate} independently $c\eps^{-1}\log\log n$ times, and take the minimum
estimate. By Markov $\Pr[X \leq (1+\eps)t] > \eps/2$. 
With probability at least $1-1/\log^2n$, the minimum estimate
is at most $(1+\eps)t$.

If the estimate is larger than the current guess of $t$, we halve the guess for $t$.
Eventually, we reach an appropriate guess where the estimates of \Thm{estimate}
to kick in. At the stage,
the minimum of $c\eps^{-1}\log\log n$ estimates is at most $(1+\eps)t$
and at least $(1-\eps)t$ with probability at least $1-1/\log n$ (we rescale $\eps$ from \Thm{estimate}). A union bound
over all errors completes the proof.
\end{proof}

\section{A $(1+\eps)$-approximation for the average degree} \label{sec:2app}

We impose a total order $\prec$ on the vertices, where $u \prec v$ if 
$d_u < d_v$ or, $d_u=d_v$ and the ID of $u$ is less than that of $v$.
(The latter is just an arbitrary but consistent tie-breaking rule).
We use $d^+_v$ to be denote the number of neighbors of $v$
``higher" than $v$ according to $\prec$.

Define random variable $X$ as follows. Pick u.a.r. vertex $v$, then pick u.a.r. vertex $u \in \Gamma(v)$. If $v \prec u$, $X = 2d_v$,
else $X = 0$. We use $\davg$ to denote $2m/n$, the average degree.

\begin{theorem} \label{thm:app} Suppose $s \geq [c\log(\eps^{-1}\log n)\mulerr^{-2}] n/\sqrt{\eps m}$.
Let $X_1, \ldots, X_s$ be i.i.d. as $X$ (defined above), and $\bar{X} = \sum_i X_i/s$.
Then $\EX[\bar{X}] = \davg$ and $\Pr[\bar{X} < (1-\eps)\davg] < 1/\eps^2\log^2n$.
\end{theorem}

\begin{proof} The expectation bound is direct. $\EX[X] = n^{-1} \sum_v (d^+_v/d_v)\cdot 2d_v = 2\sum_v d^+_v/n = \davg$.
Let $v_1, v_2, \ldots, v_s$ be the u.a.r. vertices corresponding
to $X_1, \ldots, X_s$. Set $k = \sqrt{\eps m}$, so let $S_k$ denote the highest $k$
vertices, according to $\prec$. First, some properties of $S_k$. 
For all $v \in S_k$, $d^+_v \leq k$. On the other hand, for all $v \notin S_k$,
$d_v \leq 2m/k$. (If not, the total sum of degrees in $S_k$ exceeds $2m$.)

Define random variable $Y_i$ as:
$Y_i = X_i$ if $v_i \notin S_k$ and $0$ otherwise. Denoting $\bar{Y} = \sum_i Y_i/s$,
note that $\bar{X} \geq \bar{Y}$, so it suffices to provide a lower tail for $\bar{Y}$.
Observe that $\EX[Y] = 2n^{-1}\sum_{v \notin S_k} d^+_v$
$ = 2n^{-1}(m - \sum_{v \in S_k} d^+_v) \geq 2n^{-1}(m - k^2)$ $\geq (1-\eps)\davg$.

We apply a standard multiplicative Chernoff bound on $\bar{Y} = \sum_i Y_i/s$,
noting that $Y_i \in [0,2m/k]$. So
$\Pr[\bar{Y} \geq (1-\eps)\EX[\bar{Y}]] \leq \exp(-\eps^2s\EX[\bar{Y}]/3(2m/k))$.
and going through the motions,
$$\eps^2s\EX[\bar{Y}]/3(2m/k) \geq c(1-\eps)\log(\eps^{-1}\log n)\davg n/6m = \Omega(\log(\eps^{-1}\log n))$$
By a sufficiently large choice of $c$, $\Pr[\bar{Y} \geq (1-\eps)^2\davg] \leq 1/\eps^3\log^2n$. Rescale
$\eps$ to complete the proof.
\end{proof}

To get an upper tail bound, we simply use Markov on $\bar{X}$. Of course, we do not have
a bonafide algorithm, since the value of $m$ is required to choose $s$.
But a simple geometric search for $m$ wraps up the whole proof. This is a repeat of the proof of \Thm{maintri}.

\begin{theorem} \label{thm:maindeg} There exists a $O(\eps^{-3.5}\log(\eps^{-1}\log n) n/\sqrt{m})$ algorithm
that provides an estimate in $[(1-\eps)\davg, (1+\eps)\davg]$ with probability $> 5/6$.
\end{theorem}

\begin{proof} We perform a geometric search for $m$, by guessing its value as $n^2, n^2/2, n^2/2^2, \ldots$.
Consider the estimate of \Thm{app} (for some guessed value of $m$). By Markov, $\Pr[\bar{X} \leq (1+\eps)\davg] > \eps/2$.
Basically, we run the procedure of \Thm{app} independently $c\eps^{-1}\log\log n$ times
and take the minimum estimate. With probability at least $1-1/\log^2n$, the estimate
is at most $(1+\eps)\davg$.

This yields an estimate for $m$ as well. If the estimate is larger than the current guess, we halve the guess for $m$.
Eventually, we reach a guess such that $s$ is large enough (in \Thm{app}). At the stage,
the minimum of $c\eps^{-1}\log\log n$ estimates is at most $(1+\eps)\davg$
and at least $(1-\eps)\davg$ with probability at least $1-1/\log n$. A union bound
over all errors completes the proof.
\end{proof}

We note that this proof is very similar to Feige's $(2+\eps)$-approximation. His proof is more involved
and uses tighter arguments to get the dependence on $\eps$ down to $\eps^{-1}$.

\bibliographystyle{alpha}
\bibliography{triangles}

\newcommand{\etalchar}[1]{$^{#1}$}
\begin{thebibliography}{PTTW13}

\bibitem[ADNK14]{AhDuNe+14}
N.~K. Ahmed, N.~Duffield, J.~Neville, and R.~Kompella.
\newblock Graph sample and hold: A framework for big graph analytics.
\newblock In {\em Conference on Knowledge Discovery and Data Mining (KDD)},
  2014.

\bibitem[AGM12]{AhGuMc12}
K.~J. Ahn, S.~Guha, and A.~McGregor.
\newblock Graph sketches: sparsification, spanners, and subgraphs.
\newblock In {\em Principles of Database Systems}, pages 5--14, 2012.

\bibitem[AKM12]{ArKhMa12}
S.~M. Arifuzzaman, M.~Khan, and M.~Marathe.
\newblock Patric: A parallel algorithm for counting triangles and computing
  clustering coefficients in massive networks.
\newblock Technical Report 12-042, NDSSL, 2012.

\bibitem[Avr10]{Av10}
H.~Avron.
\newblock Counting triangles in large graphs using randomized matrix trace
  estimation.
\newblock In {\em KDD workshon Large Scale Data Mining}, 2010.

\bibitem[AYZ97]{AlYuZw97}
N.~Alon, R.~Yuster, and U.~Zwick.
\newblock Finding and counting given length cycles.
\newblock {\em Algorithmica}, 17:354--364, 1997.

\bibitem[BBCG08]{BeBoCaGi08}
L.~Becchetti, P.~Boldi, C.~Castillo, and A.~Gionis.
\newblock Efficient semi-streaming algorithms for local triangle counting in
  massive graphs.
\newblock In {\em Knowledge Data and Discovery (KDD)}, pages 16--24, 2008.

\bibitem[BFL{\etalchar{+}}06]{BuFrLeMaSo06}
L.~S. Buriol, G.~Frahling, S.~Leonardi, A.~Marchetti-Spaccamela, and C.~Sohler.
\newblock Counting triangles in data streams.
\newblock In {\em Principles of Database Systems}, pages 253--262, 2006.

\bibitem[BHLP11]{BerryHLP11}
J.~W. {Berry}, B.~{Hendrickson}, R.~A. {LaViolette}, and C.~A. {Phillips}.
\newblock {Tolerating the Community Detection Resolution Limit with Edge
  Weighting}.
\newblock {\em Physical Review E}, 83(5), May 2011.

\bibitem[BPWZ14]{BjPa+14}
A.~Bj\"{o}rklund, R.~Pagh, V.~Vassilevska Williams, and U.~Zwick.
\newblock Listing triangles.
\newblock In {\em Proceedings of International Colloquium on Automata,
  Languages, and Programming (ICALP)}, pages 223--234, 2014.

\bibitem[Bur04]{Burt04}
R.~S. Burt.
\newblock Structural holes and good ideas.
\newblock {\em American Journal of Sociology}, 110(2):349--399, 2004.

\bibitem[BYKS02]{BaKuSi02}
Z.~Bar-Yossef, R.~Kumar, and D.~Sivakumar.
\newblock Reductions in streaming algorithms, with an application to counting
  triangles in graphs.
\newblock In {\em Symposium of Discrete Algorith,s}, pages 623--632, 2002.

\bibitem[CC11]{ChCh11}
S.~Chu and J.~Cheng.
\newblock Triangle listing in massive networks and its applications.
\newblock In {\em Knowledge Data and Discovery (KDD)}, pages 672--680, 2011.

\bibitem[CN85]{ChNi85}
N.~Chiba and T.~Nishizeki.
\newblock Arboricity and subgraph listing algorithms.
\newblock {\em SIAM J. Comput.}, 14:210--223, 1985.

\bibitem[Col88]{Co88}
J.~S. Coleman.
\newblock Social capital in the creation of human capital.
\newblock {\em American Journal of Sociology}, 94:S95--S120, 1988.

\bibitem[ELR15]{EdLe+15}
T.~Eden, A.~Levi, and D.~Ron.
\newblock Approximately counting triangles in sublinear time.
\newblock Technical Report TR15-046, ECCC, 2015.

\bibitem[EM02]{EcMo02}
J.-P. Eckmann and E.~Moses.
\newblock Curvature of co-links uncovers hidden thematic layers in the {World
  Wide Web}.
\newblock {\em Proceedings of the National Academy of Sciences (PNAS)},
  99(9):5825--5829, 2002.

\bibitem[Fei06]{Fe02}
U.~Feige.
\newblock Sums of independent random variables with unbounded variance and
  estimating the average degree in a graph.
\newblock {\em SIAM J. Comput.}, 35(4):964--984, 2006.

\bibitem[GR08]{GoRo08}
O.~Goldreich and D.~Ron.
\newblock Approximating average parameters of graphs.
\newblock {\em Random Structures and Algorithms}, 32(4):473--493, 2008.

\bibitem[GRS11]{GoRo+11}
M.~Gonen, D.~Ron, and Y.~Shavitt.
\newblock Counting stars and other small subgraphs in sublinear-time.
\newblock {\em SIAM J. Discrete Math}, 25(3):1365--1411, 2011.

\bibitem[HL70]{HoLe70}
P.~W. Holland and S.~Leinhardt.
\newblock A method for detecting structure in sociometric data.
\newblock {\em American Journal of Sociology}, 76:492--513, 1970.

\bibitem[IR78]{ItRo78}
A.~Ital and M.~Rodeh.
\newblock Finding a minimum circuit in a graph.
\newblock {\em SIAM Journal on Computing}, 7:413--423, 1978.

\bibitem[JG05]{JoGh05}
H.~Jowhari and M.~Ghodsi.
\newblock New streaming algorithms for counting triangles in graphs.
\newblock In {\em Computing and Combinatorics Conference (COCOON)}, pages
  710--716, 2005.

\bibitem[JSP13]{JhSePi13}
M.~Jha, C.~Seshadhri, and A.~Pinar.
\newblock A space efficient streaming algorithm for triangle counting using the
  birthday paradox.
\newblock In {\em Proceedings of the 19th ACM SIGKDD international conference
  on Knowledge discovery and data mining}, KDD '13, pages 589--597, New York,
  NY, USA, 2013. ACM.

\bibitem[KMPT10]{KoMiPeTs10}
M.~N. Kolountzakis, G.~L. Miller, R.~Peng, and C.~Tsourakakis.
\newblock Efficient triangle counting in large graphs via degree-based vertex
  partitioning.
\newblock In {\em WAW'10}, 2010.

\bibitem[KMSS12]{KaMeSaSu12}
D.~M. Kane, K.~Mehlhorn, T.~Sauerwald, and H.~Sun.
\newblock Counting arbitrary subgraphs in data streams.
\newblock In {\em International Colloquium on Automata, Languages, and
  Programming (ICALP)}, pages 598--609, 2012.

\bibitem[MSI{\etalchar{+}}02]{Milo2002}
R.~Milo, S.~{Shen-Orr}, S.~Itzkovitz, N.~Kashtan, D.~Chklovskii, and U.~Alon.
\newblock Network motifs: Simple building blocks of complex networks.
\newblock {\em Science}, 298(5594):824--827, 2002.

\bibitem[Por98]{Po98}
A.~Portes.
\newblock Social capital: Its origins and applications in modern sociology.
\newblock {\em Annual Review of Sociology}, 24(1):1--24, 1998.

\bibitem[PTTW13]{PaTaTi+13}
A.~Pavan, K.~Tangwongsan, S.~Tirthapura, and K.-L. Wu.
\newblock Counting and sampling triangles from a graph stream.
\newblock In {\em International Conference on Very Large Databases (VLDB)},
  2013.

\bibitem[SKP12]{SeKoPi11}
C.~Seshadhri, T.~G. Kolda, and A.~Pinar.
\newblock Community structure and scale-free collections of {Erd\"os-R\'enyi}
  graphs.
\newblock {\em Physical Review E}, 85(5):056109, May 2012.

\bibitem[SPK13]{SePiKo13}
C.~Seshadhri, A.~Pinar, and T.~G. Kolda.
\newblock Fast triangle counting through wedge sampling.
\newblock In {\em Proceedings of the SIAM Conference on Data Mining}, 2013.

\bibitem[SV11]{SuVa11}
S.~Suri and S.~Vassilvitskii.
\newblock Counting triangles and the curse of the last reducer.
\newblock In {\em World Wide Web (WWW)}, pages 607--614, 2011.

\bibitem[SW05a]{ScWa05-2}
T.~Schank and D.~Wagner.
\newblock Approximating clustering coefficient and transitivity.
\newblock {\em Journal of Graph Algorithms and Applications}, 9:265--275, 2005.

\bibitem[SW05b]{ScWa05}
T.~Schank and D.~Wagner.
\newblock Finding, counting and listing all triangles in large graphs, an
  experimental study.
\newblock In {\em Experimental and Efficient Algorithms}, pages 606--609.
  Springer Berlin / Heidelberg, 2005.

\bibitem[TKM11]{TsKoMi11}
C.~Tsourakakis, M.~N. Kolountzakis, and G.~Miller.
\newblock Triangle sparsifiers.
\newblock {\em J. Graph Algorithms and Applications}, 15:703--726, 2011.

\bibitem[TKMF09]{TsKaMiFa09}
C.~Tsourakakis, U.~Kang, G.~Miller, and C.~Faloutsos.
\newblock Doulion: counting triangles in massive graphs with a coin.
\newblock In {\em Knowledge Data and Discovery (KDD)}, pages 837--846, 2009.

\bibitem[TPT13]{TaPaTi13}
K.~Tangwongsan, A.~Pavan, and S.~Tirthapura.
\newblock Parallel triangle counting in massive streaming graphs.
\newblock In {\em ACM Conference on Information \& Knowledge Management
  (CIKM)}, 2013.

\bibitem[Tso08]{Ts08}
C.~Tsourakakis.
\newblock Fast counting of triangles in large real networks without counting:
  Algorithms and laws.
\newblock In {\em International Conference on Data Mining (ICDM)}, pages
  608--617, 2008.

\bibitem[WDC10]{FoDeCo10}
B.~F. Welles, A.~Van Devender, and N.~Contractor.
\newblock Is a friend a friend?: {Investigating} the structure of friendship
  networks in virtual worlds.
\newblock In {\em CHI-EA'10}, pages 4027--4032, 2010.

\end{thebibliography}

\end{document}